\documentclass[english,3p]{revtex4}
\usepackage[utf8]{inputenc}
\usepackage[stable]{footmisc}
\usepackage{amsmath,amsthm,amssymb}
\usepackage{youngtab}
\usepackage{mathtools}
\usepackage[all]{xy}
\usepackage{dsfont}
\usepackage{multirow}
\usepackage{rotating}
\usepackage{appendix}

\renewcommand{\>}{\rangle}
\newcommand\be{\begin{equation}}
\newcommand\ee{\end{equation}}
\newcommand\ot{\otimes}

\newcommand\bea{\begin{array}}
\newcommand\eea{\end{array}}
\newcommand\ben{\begin{eqnarray}}
\newcommand\een{\end{eqnarray}}

\newcommand\bei{\begin{itemize}}
\newcommand\eei{\end{itemize}}
\newcommand\bee{\begin{enumerate}}
\newcommand\eee{\end{enumerate}}

\def\ot{\otimes}

\def\bei{\begin{itemize}}
\def\eei{\end{itemize}}

\newtheorem{theorem}{Theorem}
\newtheorem{fact}[theorem]{Fact}
\newtheorem{lemma}[theorem]{Lemma}

\newtheorem{notation}[theorem]{Notation}

\newtheorem{observation}[theorem]{Observation}

\newtheorem*{rep@theorem}{\rep@title}
\newcommand{\newreptheorem}[2]{%
\newenvironment{rep#1}[1]{%
 \def\rep@title{#2 \ref{##1} (restatement)}%
 \begin{rep@theorem}}%
 {\end{rep@theorem}}}
\makeatother

\newreptheorem{theorem}{Theorem}
\newreptheorem{lemma}{Lemma}
\newreptheorem{definition}{Definition}

{\par\addvspace{\medskipamount}\noindent\textbf{Examples.}\hspace{1ex}}%
{\par\medskip}

\DeclareUnicodeCharacter{981}{\phi}
\DeclareUnicodeCharacter{348}{{\hat{S}}}
\DeclareUnicodeCharacter{8224}{{\dagger}}

\usepackage{color}
\begin{document}
\title{Construction and properties of a class of private states in arbitrary dimensions}

\author{Adam Rutkowski$^{1,2,*}$, Micha{\l} Studzi\'nski$^{1,2}$, Piotr \'Cwikli\'nski$^{1,2}$ and Micha{\l} Horodecki$^{1,2}$}
\affiliation{
$^1$ Institute of Theoretical Physics and Astrophysics, University of Gda\'nsk, 80-952 Gda\'nsk, Poland \\
$^2$ National Quantum Information Centre of Gda\'nsk, 81-824 Sopot, Poland
}

\date{\today}

\begin{abstract}
We present a construction of quantum states in dimension $d$ that has at least 1 dit of ideal key, called private dits (pdits), which covers most of the known examples of private bits (pbits) $d=2$. We examine properties of this class of states, focusing mostly on its distance to the set of separable states $\mathcal{SEP}$, showing that for a fixed dimension of key part $d_k$ the distance increases with $d_s$. We provide explicit examples of PPT states (in $d$ dimensions) which are nearly as far from separable ones as possible. Precisely, the distance from the set of $\mathcal{SEP}$ is $2 - \epsilon$, where $d$ scales with $\epsilon$ as $d \propto 1/\epsilon^3$, as opposed to $d \propto 2^{(log(4/\epsilon))^2}$ obtained in [Badzi\c{a}g et al., Phys. Rev. A 90, 012301 (2014)]. We do not use boosting (taking many copies of pdits to boost the distance) as in Badzi\c{a}g et al. paper. \end{abstract}


\maketitle
\let\oldthefootnote\thefootnote
\renewcommand{\thefootnote}{\fnsymbol{footnote}}
\footnotetext[1]{email: \url{fizar@ug.edu.pl}}
\let\thefootnote\oldthefootnote

\section{Introduction}
Quantum cryptography allows  perfect secrets sharing among  honest parties and is, up-to-date, the most successful and commercial branch among quantum information science. In 2007, quantum cryptography has been used to secure part of the vote counting in a referendum in the canton of Geneva and in 2010, in collaboration with the University of Kwazulu-Natal, South Africa, to encrypt a connection in the Durban stadium during the football World Cup in 2010. But, what is the source of its power? Briefly speaking, the fundamental
property which guarantees security of the quantum cryptography is that
if one does not know the state of a qubit, then with a high probability one disturbs
the state while trying to get to know it.

This implies there is a clear relation between quantum security and correlations in the form of quantum entanglement. If such correlations are maximal, between two qubits, they can be changed via measurement
into one bit of a secret key (also called 'classical' key). First protocols of quantum key distribution were  based only on pure entangled states \cite{BB84, Ekert91, BB84_exp} as well as, security proof \cite{ShorPreskill}, which have led to natural
expectations that pure entangled quantum states are the only source of quantum
security \cite{GisinWolf_linking, WolfGisin}.
However, we know that entanglement can be manifested not only in a pure form, but also in a mixed one. What is more, there are some mixed entangled
quantum states from which no pure entangled states can be obtained using local operations and classical communication ($LOCC$), called bound entangled states \cite{bound, refId0}. It was hoped that bound states are useless for quantum cryptography - no key would be distillable from the classical distribution. But, the quite surprising at that time, discovery of private bound entangled states, has tempered those hopes and demonstrated a clear distinction between secrecy and bound entanglement \cite{pptkey}.

The key ingredient in showing that distinction was the notion of private states (introduced in \cite{pptkey}), quantum states that contain
directly accessible, ideally secure classical key, and
private bits, p-bit - or more generally a private dit, pdit
 - which is a delocalized maximally entangled state
that still retains some entanglement monogamy result. A
quantum p-dit is composed from a $d \otimes d$ $AB$ part called "key", and
$A'B'$ called "shield", shared between Alice (subsystems
$AA'$) and Bob (subsystems $BB'$) in such a way that the
local von Neumann measurements on the key part in a
particular basis will make its results completely statistically
uncorrelated from the results of any measurement
of an eavesdropper Eve on her subsystem $E$, which is
a part of the purification $|\Psi\>_{ABA'B'E}$ of the p-dit state $\rho_{ABA'B'}$. Pdits (especially pbits), have been studied extensively for some time \cite{Pankowski2011_pbit, AugusiakH2008-multi, smallkey, keyhuge, PhysRevLett.106.030501, PhysRevA.85.012330}.

Quite recently, an important discovery has been made in studies between security and correlations. In \cite{Ozols2014_pbit}, a clean classical analogue of bound entanglement and private bound entanglement has been provided, where the authors have constructed private bound entangled states based on unambiguous classical probability distribution to a quantum state that is not based on a "standard" key/shield scheme, opening a new direction in studies of private states.

Our paper is organized in the following way. In Sec.~\ref{II} we present a general construction of the new class of pdits and show that for specific choices of parameters we can reduce this new class to the cases previously known in the literature. In Sec.~\ref{III} we investigate properties of the new set of pdits. Namely we calculate the trace distance of arbitrary pdits from the new class from the pdit in maximally entangled form (Lemma~\ref{thm:main}). We also show that for the specific subclass this distance scales inversely with the dimension of the shield part $d_s$ (Lemma~\ref{depOnds}). At the end of this section, we give the lower bound for the trace distance from the set of separable states $\mathcal{SEP}$ and our subclass (Lemma~\ref{distSep}) which gives better estimation than the previous one~\cite{Badziag_pbits}.
What is the most important, we are able to show that for particular subclass of pdits, we do not need to take many copies of pdits to boost the distance from the set of separable set $\mathcal{SEP}$ (like in~\cite{Badziag_pbits}) using our construction. We also show that our family of states  approximate the set of separable states obtaining the distance equal to $2 - \epsilon$ and improving the scaling of $\epsilon$ with the distance. Additionally, we present two appendices in which we  describe a special method which allows us to prove one of the crucial statements in our paper, i.e. Lemma~\ref{depOnds} (Appendix~\ref{tenMet}). In  Appendix~\ref{XY} we remind the special construction of the set of operators which is one of the possible realizations of operators with desired spectra needed in Sec.~\ref{III}.

\section{General construction of pdits}
\label{II}
As we have mentioned in the Introduction we want to construct a four partite state $\rho_{ABA'B'}$ (pdit) which has $\operatorname{PPT}$ property and it is close to pdits in the so called maximally entangled form (see Section~\ref{III}). Let us consider the following state:
\be
\label{1}
\rho_{ABA'B'}=\sum_{l=0}^d \omega_l \in \mathcal{B}\left(\mathcal{H}_{d_k} \otimes \mathcal{H}_{d_k} \otimes \mathcal{H}_{d_s} \otimes \mathcal{H}_{d_s} \right),
\ee
where $\mathcal{B}(\mathcal{H})$ is the algebra of all bounded linear operators on Hilbert space $\mathcal{H}$, $d=\frac{1}{2}d_k(d_k-1)$ and by $d_k$ we denote the dimension of the key part acting on $AB$ and by $d_s$ the dimension of the shield part acting on $A'B'$. Now we describe each of the components from Eq.~\eqref{1}. First of all, we define the term $\omega_0$ as:
\be
\label{2}
\omega_0=\sum_{i,j=0}^{d_k-1}|i\>\<j|\otimes |i\>\<j| \otimes a_{ij}^{(0,0)},
\ee
where every $a_{ij}^{(0,0)}\in\mathcal{B}\left(\mathcal{H}_{d_s}\otimes \mathcal{H}_{d_s}\right)$. From now, every matrix of the form~\eqref{2} we will call matrix in the maximally entangled form.
The rest of elements $\omega_l$, for $1\leq l \leq \frac{1}{2}d_k(d_k-1)$ from Eq.~\eqref{1} are given by the following formula
\be
\label{3}
\begin{split}
\omega_l&=|i\>\<i|\otimes |j\>\<j| \otimes a_{00}^{(i,j)}+|i\>\<j|\otimes |j\>\<i| \otimes a_{01}^{(i,j)}+\\ &+|j\>\<i|\otimes |i\>\<j|\otimes a_{10}^{(i,j)}+|j\>\<j|\otimes |i\>\<i|\otimes a_{11}^{(i,j)},
\end{split}
\ee
where $i,j=1,\ldots, d_k-1$ and $i<j$. In the above we also implicitly assume bijection function between indices $l$ and $i,j$.

Let us introduce the following notation, namely:
\be
\label{ij}
A^{(i,j)}=\begin{pmatrix} a^{(i,j)}_{00} & a^{(i,j)}_{01}\\ a^{(i,j)}_{10} & a^{(i,j)}_{11}  \end{pmatrix},
\ee
where $i,j=0,\ldots, d_k-1$ for $i<j$. Separately, for the term $A^{(0,0)}$, we have
\be
\label{00}
A^{(0,0)}=\begin{pmatrix} a_{00}^{(0,0)} & \cdots & a^{(0,0)}_{0,d_k-1}\\ \vdots & \ddots & \vdots \\ a^{(0,0)}_{d_k-1,0} & \cdots & a^{(0,0)}_{d_k-1,d_k-1} \end{pmatrix}.
\ee
Then, there is an explicit connection between positivity of the state $\rho_{ABA'B'}$ and each submatrix $ A^{(i,j)}$ and positivity of $\rho_{ABA'B'}^{\operatorname{T}_{A'}\operatorname{T}_{B'}}$ and each block $A^{(i,j)}$ after partial transposition on the system $B'$. This can be summarized as follows
\begin{observation}
\label{obs1}
We have the following relations between positivity of the state $\rho_{ABA'B'}$ before and after partial transposition and positivity properties of every block $A^{(i,j)}$:
\begin{enumerate}
\item Positivity of the state $\rho_{ABA'B'}$
\be
\rho_{ABA'B'}\geq 0 \Leftrightarrow A^{(i,j)}\geq 0,
\ee
\item Positivity of the state $\rho_{ABA'B'}$ with respect to partial transposition in the cut $AB:A'B'$
\be
\left( \text{\(\mathds{1}\)}_{A} \otimes \operatorname{T}_B \otimes \text{\(\mathds{1}\)}_{A'} \otimes \operatorname{T}_{B'}\right) \rho_{ABA'B'} \geq 0 \Leftrightarrow \widetilde{A}^{(i,j)}\geq 0,
\ee
where $\widetilde{A}^{(i,j)}$  is given by
\[
\tilde{A}^{(i,j)}=\left(\begin{array}{cc}
\tilde{a}_{00}^{(i,j)} & \tilde{a}_{ij}^{(0,0)}\\
\tilde{a}_{ji}^{(0,0)} & \tilde{a}_{11}^{(i,j)}
\end{array}\right),\quad i,j=0,\ldots,d_{k}-1 \ \text{with} \  i<j,
\]
and
\[
\tilde{A}_{ij}^{(0,0)}=\begin{cases}
\tilde{a}_{ij}^{(0,0)},\qquad & i=j\\
\tilde{a}_{01}^{(i,j)},\qquad & i<j\\
\tilde{a}_{10}^{(i,j)},\qquad & i>j
\end{cases},\quad i,j=0,\ldots,d_{k}-1.
\]
In the above, we have $\widetilde{a}_{00}^{(i,j)}=\left(\text{\(\mathds{1}\)}_{B} \otimes \operatorname{T}_{B'} \right)a_{00}^{(i,j)}$ and so on.
\end{enumerate}
\end{observation}

\begin{proof}
The proof of above statement is based on straightforward observation. Namely one can notice that every component of the state from equation~\eqref{1} is defined on different subspaces which are orthogonal to each other, thus every block can be treated separately - we can consider positivity and PPT conditions on each of the component independently. This fact implies all claimed properties of states $\rho_{ABA'B'}$ from~\eqref{1}.
\end{proof}

At the end of this section we show for which choices of matrices $\omega_0$ and $\omega_l$ we can reduce our general construction, given by formulas~\eqref{1},~\eqref{2} and~\eqref{3}, to the previously known cases. First let us write general matrix expressions for state $\rho_{ABA'B'}$ from the formula~\eqref{1} when the dimension of the key part is $d_k=2,3$. Namely for $d_k=2$ we have
\be
\rho_{ABA'B'}=\omega_0+\omega_1,
\ee
where
\be
\omega_0=\left(\begin{array}{cc|cc}a_{00}^{(0,0)} & \cdot & \cdot & a_{01}^{(0,0)}\\
\cdot & \cdot & \cdot & \cdot \\
\hline
\cdot & \cdot & \cdot & \cdot \\
a_{10}^{(0,0)} & \cdot & \cdot & a_{11}^{(0,0)} \end{array}\right), \ \omega_1=\left(\begin{array}{cc|cc} \cdot & \cdot & \cdot & \cdot\\
\cdot & a_{00}^{(0,1)} & a_{01}^{(0,1)} & \cdot \\ \hline
\cdot & a_{10}^{(0,1)} & a_{11}^{(0,1)} & \cdot \\
\cdot & \cdot & \cdot & \cdot\end{array}\right)
\ee

For $d_k=3$ state $\rho_{ABA'B'}$ is represented as
\be
\rho_{ABA'B'}=\omega_0+\omega_1+\omega_2+\omega_3 \in \mathcal{B}\left(\mathbb{C}^3 \ot \mathbb{C}^3 \ot \mathbb{C}^{d_s} \ot \mathbb{C}^{d_s}\right),
\ee
where
\begin{widetext}
\be
\begin{split}
\rho_0&= \left(\begin{array}{ccc|ccc|ccc}a_{00}^{(0,0)} & \cdot & \cdot & \cdot & a_{01}^{(0,0)} & \cdot & \cdot & \cdot & a_{02}^{(0,0)}\\
\cdot & \cdot & \cdot & \cdot & \cdot & \cdot & \cdot & \cdot & \cdot\\ \cdot & \cdot & \cdot & \cdot & \cdot & \cdot & \cdot & \cdot & \cdot\\ \hline
\cdot & \cdot & \cdot & \cdot & \cdot & \cdot & \cdot & \cdot & \cdot \\
a_{10}^{(0,0)} & \cdot & \cdot & \cdot & a_{11}^{(0,0)} & \cdot & \cdot & \cdot & a_{12}^{(0,0)}\\
\cdot & \cdot & \cdot & \cdot & \cdot & \cdot & \cdot & \cdot & \cdot\\ \hline
\cdot & \cdot & \cdot & \cdot & \cdot & \cdot & \cdot & \cdot & \cdot\\
\cdot & \cdot & \cdot & \cdot & \cdot & \cdot & \cdot & \cdot & \cdot\\
a_{20}^{(0,0)} & \cdot & \cdot & \cdot & a_{21}^{(0,0)} & \cdot & \cdot & \cdot & a_{22}^{(0,0)}
\end{array} \right), \
\rho_1=\left(\begin{array}{ccc|ccc|ccc}\cdot & \cdot & \cdot & \cdot & \cdot & \cdot & \cdot & \cdot & \cdot\\
\cdot & a_{00}^{(0,1)} & \cdot & a_{01}^{(0,1)} & \cdot & \cdot & \cdot & \cdot & \cdot\\
\cdot & \cdot & \cdot & \cdot & \cdot & \cdot & \cdot & \cdot & \cdot\\ \hline
\cdot & a_{10}^{(0,1)} & \cdot & a_{11}^{(0,1)} & \cdot & \cdot & \cdot & \cdot & \cdot \\
\cdot & \cdot & \cdot & \cdot & \cdot & \cdot & \cdot & \cdot & \cdot \\
\cdot & \cdot & \cdot & \cdot & \cdot & \cdot & \cdot & \cdot & \cdot\\ \hline
\cdot & \cdot & \cdot & \cdot & \cdot & \cdot & \cdot & \cdot & \cdot\\
\cdot & \cdot & \cdot & \cdot & \cdot & \cdot & \cdot & \cdot & \cdot\\
\cdot & \cdot & \cdot & \cdot & \cdot & \cdot & \cdot & \cdot & \cdot
\end{array} \right), \
\rho_2=\left(\begin{array}{ccc|ccc|ccc}\cdot & \cdot & \cdot & \cdot & \cdot & \cdot & \cdot & \cdot & \cdot\\
\cdot & \cdot & \cdot & \cdot & \cdot & \cdot & \cdot & \cdot & \cdot\\
\cdot & \cdot & a_{00}^{(0,2)} & \cdot & \cdot & \cdot & a_{01}^{(0,2)} & \cdot & \cdot\\ \hline
\cdot & \cdot & \cdot & \cdot & \cdot & \cdot & \cdot & \cdot & \cdot \\
\cdot & \cdot & \cdot & \cdot & \cdot & \cdot & \cdot & \cdot & \cdot \\
\cdot & \cdot & \cdot & \cdot & \cdot & \cdot & \cdot & \cdot & \cdot\\ \hline
\cdot & \cdot & a_{10}^{(0,2)} & \cdot & \cdot & \cdot & a_{11}^{(0,2)} & \cdot & \cdot\\
\cdot & \cdot & \cdot & \cdot & \cdot & \cdot & \cdot & \cdot & \cdot\\
\cdot & \cdot & \cdot & \cdot & \cdot & \cdot & \cdot & \cdot & \cdot
\end{array} \right),\\
\rho_3&=\left(\begin{array}{ccc|ccc|ccc}\cdot & \cdot & \cdot & \cdot & \cdot & \cdot & \cdot & \cdot & \cdot\\
\cdot & \cdot & \cdot & \cdot & \cdot & \cdot & \cdot & \cdot & \cdot\\
\cdot & \cdot & \cdot & \cdot & \cdot & \cdot & \cdot & \cdot & \cdot\\ \hline
\cdot & \cdot & \cdot & \cdot & \cdot & \cdot & \cdot & \cdot & \cdot \\
\cdot & \cdot & \cdot & \cdot & \cdot & \cdot & \cdot & \cdot & \cdot \\
\cdot & \cdot & \cdot & \cdot & \cdot & a_{00}^{(1,2)} & \cdot & a_{01}^{(1,2)} & \cdot\\ \hline
\cdot & \cdot & \cdot & \cdot & \cdot & \cdot & \cdot & \cdot & \cdot\\
\cdot & \cdot & \cdot & \cdot & \cdot & a_{10}^{(1,2)} & \cdot & a_{11}^{(1,2)} & \cdot\\
\cdot & \cdot & \cdot & \cdot & \cdot & \cdot & \cdot & \cdot & \cdot
\end{array} \right).
\end{split}
\ee
\end{widetext}
Form the above examples we see that operators $\omega_k$ are supported on orthogonal subspaces. Now, we are ready to present five examples of private states which belong to our class:
\begin{enumerate}
\item Suppose that $\gamma^{\operatorname{V}} \in \mathcal{B}\left(\mathbb{C}^2 \otimes \mathbb{C}^2 \otimes \mathbb{C}^{d_s} \otimes \mathbb{C}^{d_s}\right)$ such that
\be
\gamma^{\operatorname{V}}=\frac{1}{2}\begin{pmatrix}\text{\(\mathds{1}\)}/d_s^2 & \cdot & \cdot & \operatorname{V}/d_s^2\\
\cdot & \cdot & \cdot & \cdot\\
\cdot & \cdot & \cdot & \cdot\\
\operatorname{V}/d_s^2 & \cdot & \cdot & \text{\(\mathds{1}\)}/d_s^2
\end{pmatrix},
\ee
where $\operatorname{V}=\sum_{i=0}^{d_s-1}|ij\>\<ji|$ is known as the swap operator, $\text{\(\mathds{1}\)}$ is the identity matrix of dimension $d_s^2 \times d_s^2$ and by dots we denote matrices of dimension $d_s^2 \times d_s^2$ filled with zeros~\cite{KH_phd}.
\item Suppose that $\rho_{\text{flower}} \in \mathcal{B}\left(\mathbb{C}^2 \otimes \mathbb{C}^2 \otimes \mathbb{C}^{d_s} \otimes \mathbb{C}^{d_s}\right)$ such that
\be
\label{flower}
\rho_{\text{flower}}=\frac{1}{2}\begin{pmatrix} \sigma & \cdot & \cdot & U^{\operatorname{T}}/d_s\\
\cdot & \cdot & \cdot & \cdot\\
\cdot & \cdot & \cdot & \cdot\\
U^*/d_s & \cdot & \cdot & \sigma
\end{pmatrix},
\ee
where $\sigma=(1/d_s)\sum_{i=0}^{d_s-1}|ii\>\<ii|$ is the classical maximally correlated state and $U$ is an embedding of unitary transformation $W=\sum_{i,j=0}^{d_s-1}w_{ij}|i\>\<j|$ in the form $U=\sum_{i,j=0}^{d_s-1}w_{ij}|ii\>\<jj|$. The state~\eqref{flower} is known as the flower state \cite{keyhuge}.

\item Suppose that $\rho \in \mathcal{B}\left(\mathbb{C}^2 \otimes \mathbb{C}^2 \otimes \mathbb{C}^{ld_s} \otimes \mathbb{C}^{ld_s}\right)$ such that
\be
\label{key}
\rho_{ABA'B'}=\frac{1}{2}\begin{pmatrix} p(\tau_0+\tau_1) & \cdot & \cdot & p(\tau_1-\tau_0)\\
\cdot & (1-2p)\tau_0 & \cdot & \cdot \\
\cdot & \cdot & (1-2p)\tau_0 & \cdot \\
p(\tau_1-\tau_0) & \cdot & \cdot & p(\tau_0+\tau_1)
 \end{pmatrix}.
\ee
In the above $\tau_0=\rho_s^{\ot l}, \tau_1=[(\rho_a+\rho_s)/2]^{\ot l}$, $l$ is a positive integer number, $\mathcal{B}(\mathbb{C}^{d_s})\ni\rho_s=\frac{2}{d_s^2+d_s}\operatorname{P}_{sym}$, $\mathcal{B}(\mathbb{C}^{d_s})\ni\rho_a=\frac{2}{d_s^2-d_s}\operatorname{P}_{as}$, where $\operatorname{P}_{sym}, \operatorname{P}_{as}$ are respectively symmetric and antisymmetric projectors  for bipartite case. It has been shown that a class of states~\eqref{key} is bound entangled with a private key $K_D >0$ \cite{keyhuge}.

\item Finally, let us take $\rho_{ABA'B'} \in \mathcal{B}\left(\mathbb{C}^2 \otimes \mathbb{C}^2 \otimes \mathbb{C}^{d_s} \otimes \mathbb{C}^{d_s}\right)$ in the most general form of pbit, the so-called $X-$form of pbit \cite{keyhuge}:
\be
\rho_{ABA'B'}=\frac{1}{2}\begin{pmatrix}\sqrt{XX^{\dagger}} & \cdot & \cdot & X\\ \cdot & \cdot & \cdot & \cdot \\ \cdot & \cdot & \cdot & \cdot \\ X^{\dagger} & \cdot & \cdot & \sqrt{X^{\dagger}X}\end{pmatrix},
\ee
where $X$ is an arbitrary operator with $||X||_1=1$ and dots represent zero matrices.
\item For a larger dimension of the key part, for example $d_k=3$, we can take $\rho_{ABA'B'} \in \mathcal{B}\left(\mathbb{C}^3 \otimes \mathbb{C}^3 \otimes \mathbb{C}^{d_s} \otimes \mathbb{C}^{d_s}\right)$ in the following way
\be
\rho_{ABA'B'}=\frac{1}{3}\begin{pmatrix}\sqrt{XX^{\dagger}} & \cdot & \cdot & \cdot & X & \cdot & \cdot & \cdot & XY\\
\cdot & \cdot & \cdot & \cdot & \cdot & \cdot & \cdot & \cdot & \cdot \\
\cdot & \cdot & \cdot & \cdot & \cdot & \cdot & \cdot & \cdot & \cdot \\
\cdot & \cdot & \cdot & \cdot & \cdot & \cdot & \cdot & \cdot & \cdot \\
X^{\dagger} & \cdot & \cdot & \cdot & \sqrt{X^{\dagger}X} & \cdot & \cdot & \cdot & Y\\
\cdot & \cdot & \cdot & \cdot & \cdot & \cdot & \cdot & \cdot & \cdot \\
\cdot & \cdot & \cdot & \cdot & \cdot & \cdot & \cdot & \cdot & \cdot \\
\cdot & \cdot & \cdot & \cdot & \cdot & \cdot & \cdot & \cdot & \cdot \\
(XY)^{\dagger} & \cdot & \cdot & \cdot & Y^{\dagger} & \cdot & \cdot & \cdot & \sqrt{Y^{\dagger}Y}
\end{pmatrix},
\ee
where matrices $X,Y$ satisfy $||X||_1=||Y||_1=1$ and $X=WY^{\dagger}$ for an arbitrary unitary transformation $W$~\cite{KH_phd}.
\end{enumerate}
From the above examples we can easily figure out explicit form of the operators $\omega_k$ in every case.

\section{Properties}
\label{III}
In this section we formulate theorem, which determines the distance in the trace norm between our set of states and the set of pdits in the maximally entangled form. Next, we show (Lemma~\ref{depOnds}) that this distance depends on the shield dimension $d_s$ for a special but quite general subclass of pdits. Namely, we show that this distance scales inversely with the shield dimension $d_s$. At the end we also calculate the trace distance from the set of separable states using a special representation of the pdit (Lemma~\ref{distSep}). In this and next sections,
without loss of generality, we assume the state $\rho_{ABA'B'}$ to be
\be
\rho_{ABA'B'}=p\gamma_0+\frac{q}{d}\sum_{i=1}^{d}\gamma_i,
\ee
where $p+q=1$, $d=\frac{1}{2}d_k(d_k-1)$ and
\be
\gamma_0=\frac{1}{\operatorname{Tr \omega_0}}\omega_0, \qquad \gamma_i=\frac{1}{\operatorname{Tr}\omega_i}\omega_i,
\ee
so such a state indeed belongs to the class defined in Sec.~\ref{II}, and state $\gamma_0$ we will call pdit in the maximally entangled form.  Now, we are ready to formulate the main results of this section.
\begin{lemma}
\label{thm:main}
Let us assume that we are given $\rho_{ABA'B'}$ as in Eq.~\eqref{1} and the pdit $\gamma_0$ in its maximally entangled form, then the following statement holds:
\be
\label{eq1:thm:main}
||\rho_{ABA'B'}-\gamma_0||_1=2q.
\ee
\end{lemma}

\begin{proof}
The proof is based on straightforward calculations.
Let us compute the desired trace distance between $\rho_{ABA'B'}$ and $\gamma_0$:
\be
\begin{split}
||\rho_{ABA'B'}-\gamma_0||_1=\left| \left|p\gamma_0 +\frac{q}{d}\sum_{i=1}^d\gamma_i-\gamma_0 \right|\right|_1
=\left| \left|\frac{q}{d}\sum_{i=1}^d\gamma_i-q\gamma_0 \right|\right|_1=\frac{q}{d}\left| \left| \sum_{i=1}^d\gamma_i-d\gamma_0 \right|\right|_1.
\end{split}
\ee
Now, using the definition of trace norm we rewrite the last term from the above calculations in a more explicit way
\be
\begin{split}
||\rho_{ABA'B'}-\gamma_0||_1=\frac{q}{d}\operatorname{Tr}\left[\left(\sum_{i=1}^d\gamma_i-d\gamma_0\right)\left(\sum_{i=1}^d\gamma_i-d\gamma_0\right)^{\dagger}\right]^{1/2},
\end{split}
\ee
because we deal with hermitian matrices we have
\be
\begin{split}
||\rho_{ABA'B'}-\gamma_0||_1=\frac{q}{d}\operatorname{Tr}\left[\left(\sum_{i=1}^d\gamma_i+d\gamma_0\right)^2 \right]^{1/2},
\end{split}
\ee
and finally
\be
||\rho_{ABA'B'}-\gamma_0||_1=\frac{q}{d}\operatorname{Tr}\left[\sum_{i=1}^d\gamma_i+d\gamma_0 \right]=2q.
\ee
We obtain the statement of our theorem, so the proof is finished.
\end{proof}
Next, we formulate and prove the next lemma, which states that the distance between our class of states given in Sec.~\ref{II} and pdit in its maximally entangled form decreases with the dimension of the shield part $d_s$. We do it for a specific choice of operators $\omega_0, \omega_k$ given by Eqs~\eqref{2},~\eqref{3}, which gives a wide class of pdits. Let us choose all matrices $a_{ij}^{(0,0)}=a$, where $0\leq i,j \leq d_k$  in such a way that
\be
\label{a1}
\operatorname{spec}(a)=\left\{\frac{1}{d_s^2},\ldots,\frac{1}{d_s^2}\right\},
\ee
and all matrices $a^{(i,j)}_{mn}=b$, where $0\leq m,n \leq 1$ and $0\leq i,j \leq \frac{1}{2}d_k(d_k-1)$ with $i<j$ as:
\be
\label{a2}
\operatorname{spec}(b)=\left\{\frac{1}{d_s},\ldots,\frac{1}{d_s}\right\}.
\ee
We also assume that operators which have such  spectra are invariant under partial transposition with respect to the system $B'$. At this point we refer the reader to Appendix~\ref{XY} in which we show the explicit form of operators satisfying all requirements.
Using the above definitions we are ready to show the following
\begin{lemma}
\label{depOnds}
Let us consider the class of states given by
\be
\rho_{ABA'B'}=p\gamma_0+\frac{q}{d}\sum_{i=1}^{d}\gamma_i,
\label{rhobound}
\ee
where $q=1-p$, $d=\frac{1}{2}d_k(d_k-1)$ and states $\gamma_0, \gamma_i$ are given by Eqs~\eqref{2},~\eqref{3}, together with~\eqref{a1},~\eqref{a2}. Then the trace distance from the set of private dits in maximally entangled form is equal to
\be
\label{specDist}
\frac{1}{2}||\rho_{ABA'B'}-\gamma_0||_1=\frac{1}{1+\frac{d_s}{d_k-1}},
\ee
where $d_s$ is the dimension of the shield part and $d_k$ - the dimension of the key part.
\end{lemma}

\begin{proof}
 We need to show that in our scheme the parameter $q$ which is equal to the trace distance between states $\rho_{ABA'B'}$ and pdits $\gamma_0$ in their maximally entangled form is equal to $1/(1+\frac{d_s}{d_k-1})$, where $d_s, d_k$ are dimensions of the shield and the key part respectively. To prove this property we use the construction described in details in Appendix~\ref{tenMet}. Because we have assumed that our matrices $a$ and $b$ are invariant under partial transposition with respect to the system $B'$ we can directly use equality from Eq.~\eqref{equality} putting instead of $\widetilde{a}$, a matrix $a$ and instead of $\widetilde{b}$, a matrix $b$. Then we have
 \be
 \frac{q}{d_k-1}\lambda\left(b\right)-p\lambda\left(a\right)=0,
 \ee
 where by $\lambda(a), \lambda(b)$ we denote nonzero eigenvalues of operators $a$ and $b$ respectively.
 Now using formulas~\eqref{a1} and~\eqref{a2} we get
 \be
\begin{split}
\frac{q}{d_k-1}\frac{1}{d_s}-p\frac{1}{d_s^2}=0.
\end{split}
\ee
Solving the above equality with $p=1-q$ we obtain the statement of our Lemma. This finishes the proof.
\end{proof}

Before we formulate next result we introduce the following notation
\begin{notation}
Suppose that we are given a quantum state $\rho$ and the set of separable states $\mathcal{SEP}$, then by $\operatorname{dist}(\rho,\mathcal{SEP})$ we understand the following quantity
\be
\operatorname{dist}(\rho,\mathcal{SEP})=\mathop{\min}\limits_{\sigma \in \mathcal{SEP}}||\rho-\sigma||_1,
\ee
which is of course double minimal trace distance. In further part of this manuscript whenever we talk about distance we mean the above notation.
\end{notation}

Now we are ready to calculate the lower bound on distance between the set of separable states denoted by $\mathcal{SEP}$ and our subclass of pdits given in the argumentation before Lemma~\ref{depOnds}.
\begin{lemma}
\label{distSep}
The distance between set of separable states $\mathcal{SEP}$ and class of states of the form
\be
\label{s2}
\rho_{ABA'B'}=p\gamma_0+\frac{q}{d}\sum_{i=1}^d\gamma_i,
\ee
where $q=1-p$ and $d=\frac{1}{2}d_k(d_k-1)$ is bounded form below:
\be
\label{ourbound}
\operatorname{dist}(\rho_{ABA'B'},\mathcal{SEP})\geq 2-\frac{2}{d_k}-\frac{2}{1+\frac{d_s}{d_k-1}},
\ee
where $d_s$ denotes the dimension of the shield part and the $d_k$ dimension of the key part.
\end{lemma}

\begin{proof}
In our proof we use the fact that distance between an arbitrary private state $\bar{\gamma}$ and the set of separable states $\mathcal{SEP}$ is bounded from below~\cite{Badziag_pbits} by:
\be
\operatorname{dist}(\bar{\gamma},\mathcal{SEP})\geq 2-\frac{2}{d_k},
\ee
where $d_k$ is dimension of key part. Because the above bound holds for an arbitrary private state, it holds also for a pdit in its maximally entangled form $\gamma_0$. Now let us take the closest separable state $\omega$ to $\rho_{ABA'B'}$ given by Eq.~\eqref{s2}. Using the triangle inequality we can write
\be
\label{chain}
||\rho_{ABA'B'}-\omega||_1+||\rho_{ABA'B'}-{\gamma_0}||_1\geq ||\omega-\gamma_0||_1\geq \operatorname{dist}(\gamma_0, \mathcal{SEP})\geq 2-\frac{2}{d_k},
\ee
but from Lemma~\ref{depOnds} we know, that $||\rho_{ABA'B'}-\gamma_0||_1=\frac{2}{1+\frac{d_s}{d_k-1}}$, so
\be
||\rho_{ABA'B'}-\omega||_1+\frac{2}{1+\frac{d_s}{d_k-1}}\geq 2-\frac{2}{d_k}.
\ee
The above inequality directly implies that
\be
\operatorname{dist}(\rho_{ABA'B'},\mathcal{SEP})\geq 2-\frac{2}{d_k}-\frac{2}{1+\frac{d_s}{d_k-1}}.
\ee
\end{proof}
Let us notice  that for our special case $d_k=2$, when Alice and Bob share qubit states, the bound obviously improves with dimension of the shield part and has minimum for $d_s=2$, i.e. when Alice and Bob share four-qubit state.




Let us recall that the state from Lemma \ref{depOnds} can be considered as a PPT state acting on $\mathbb{C}^d \ot \mathbb{C}^d$, where $d=d_s d_k$. We can formulate the following, recovering the result from \cite{Badziag_pbits} and \cite{Beigi2010}
\begin{theorem}
For an arbitrary $\epsilon>0$ there exists a PPT state $\rho$ acting on the Hilbert space $\mathbb{C}^d \ot \mathbb{C}^d$ with $d\leq \frac{c}{\epsilon^3} $ such that:
\be
\operatorname{dist}(\rho,\mathcal{SEP})\geq 2-\epsilon,
\ee
where $c$ is constant. The sate is given by (\ref{rhobound}).
\end{theorem}
The proof is straightforward and based on simple calculations, so it is not reported here. We have found analytically that constant $c<64$. This result considerably improves  the bound obtained in \cite{Badziag_pbits}.

\section{Summary}
In this paper, we present the construction of the set of pdits, which contains many known examples of private states from the literature (Section~\ref{II}). We also present the result specifying the trace distance between our set of pdits and the pdit in the maximally entangled form. Next, we connect this result with a dimension of the shield part $d_s$, and we prove that this distance is inversely proportional to $d_s$ at least for a particular subclass of pdits.
We also calculate the trace distance from the set of separable states $\mathcal{SEP}$ and show that for a fixed dimension of key part $d_k$ this distance decreases with $d_s$. The most interesting property of our new class of states, which differs it from the known results is that we do not need many copies of them (see ~\cite{Badziag_pbits}) to boost distance from the set of separable states $\mathcal{SEP}$ (Section~\ref{III}). We also provide explicit calculations of a family of states such that we recover the $2 - \epsilon$ distance from $\mathcal{SEP}$ \cite{Beigi2010}, \cite{Badziag_pbits} in a natural and basic way. Finally, we show that the scaling of $\epsilon$ with the distance is  $d \propto 1/\epsilon^3$, and it is considerably better than  $d \propto 2^{(log(4/\epsilon))^2}$  from  \cite{Badziag_pbits} .

\section{Acknowledgments}
M.H. thanks S. Szarek for valuable  discussions. Part of this work was done in National Quantum Information Center of Gdansk. M.S. and P.\'C. thank the hospitality of IBM TJ Watson Research Center, where (another) part of this work was done and acknowledge helpful discussions with Graeme Smith and John Smolin about private states and private bits. M.S. is supported by the International PhD Project "Physics of future quantum-based information technologies": grant MPD/2009-3/4 from Foundation for Polish Sciences. M.S. and P.\'C. acknowledge the support from the National Science Centre project Maestro DEC-2011/02/A/ST2/00305. A.R is supported by a postdoc internship, decision number DEC-2012/04/S/ST2/00002, from the (Polish) National Science Center. M.H. is supported by
Polish Ministry of Science and Higher Education Grant no. IdP2011 000361.

\appendix

\section{Construction of special pdits subclass}
\label{tenMet}
In this section we describe the method which we have used to obtain explicit positivity conditions in the proof of the Lemma~\ref{depOnds} for an arbitrary dimension of the key part $d_k$. Our argumentation is made for the specific subclass of states given at the begin of Section~\ref{II}. Suppose that above mentioned subclass is in the following form
\be
\label{LAux2}
\rho_{ABA'B'}=p\gamma_0+\frac{q}{d}\sum_{i=1}^d \gamma_i\in \mathcal{B}\left(\mathcal{H}_{d_k} \otimes \mathcal{H}_{d_k} \otimes \mathcal{H}_{d_s} \otimes \mathcal{H}_{d_s}\right),
\ee
where $d=\frac{1}{2}d_k(d_k-1)$ and  matrices $\gamma_0$, $\gamma_i$ are defined on orthogonal subspaces in the same similar way as in~\eqref{2},~\eqref{3}.
Of course to satisfy $\rho_{ABA'B'} \geq 0$ we need $\gamma_0 \geq 0$ and $\gamma_i \geq 0$.  In our construction operator $\gamma_0$ corresponds with~\eqref{2}, but all $a^{(0,0)}_{ij}=a$ together with $||a||_1=1$. Similarly we proceed for the matrices $\gamma_i$ by putting all submatrices $a^{(i,j)}_{mn}$ equal to $b$ with $||b||_1=1$. Thanks to this we have explicit connection between states $\gamma_0, \gamma_i$ and $\omega_0, \omega_i$ from~\eqref{2},~\eqref{3} by the following formulas
\be
\label{gammY}
\gamma_0=\frac{1}{d_k}\omega_0, \qquad \gamma_i=\frac{1}{2}\omega_i, \qquad \text{where} \qquad d=\frac{1}{2}d_k(d_k-1).
\ee
It is easy to see that to ensure $\operatorname{PPT}$ property respect to partial transposition on $BB'$ it is enough to satisfy $\operatorname{PPT}$ condition for every component of~\eqref{LAux2} separately after partial transposition. Thanks to this and property of orthogonality we can write
\be
\label{Aux3}
\left( \text{\(\mathds{1}\)}_{A} \otimes \operatorname{T}_B \otimes \text{\(\mathds{1}\)}_{A'} \otimes \operatorname{T}_{B'}\right) \rho_{ABA'B'} \geq 0 \Leftrightarrow \operatorname{PT}_{d_k}=\begin{pmatrix}p\widetilde{a}& \frac{q}{d_k-1}\widetilde{b} & \cdots & \frac{q}{d_k-1}\widetilde{b}\\
 \frac{q}{d_k-1}\widetilde{b} & p\widetilde{a} & \cdots & \frac{q}{d_k-1}\widetilde{b}\\
 \vdots & & \ddots & \vdots \\
 \frac{q}{d_k-1}\widetilde{b} & \cdots & p\widetilde{b} & \frac{q}{d_k-1}\widetilde{a} \end{pmatrix} \geq 0,
\ee
and
\be
\label{Aux4}
\left( \text{\(\mathds{1}\)}_{A} \otimes \operatorname{T}_B \otimes \text{\(\mathds{1}\)}_{A'} \otimes \operatorname{T}_{B'}\right) \rho_{ABA'B'} \geq 0 \Leftrightarrow \operatorname{PT}=\begin{pmatrix} \frac{q}{d_k-1}\widetilde{b} & p\widetilde{a}\\ p\widetilde{a} & \frac{q}{d_k-1}\widetilde{b}
 \end{pmatrix} \geq 0,
\ee
where $\widetilde{a}, \widetilde{b}$ are operators $a,b$ after partial transposition respect to subsystem $B'$, and second condition is taken $d_k$ times.

In general, still it is hard to say are constraints~\eqref{Aux3} and~\eqref{Aux4} satisfied, but there is nice mathematical trick which allows us to rewrite above condition in more operative way. Namely matrices $\operatorname{PT}_{d_k}$ and $\operatorname{PT}$ can be written as
\be
\begin{split}
\operatorname{PT}_{d_k}&=\text{\(\mathds{1}\)}_{d_k} \otimes p \widetilde{a} - \text{\(\mathds{1}\)}_{d_k} \otimes \frac{q}{d_k-1} \widetilde{b} + \mathbb{I}_{d_k}\otimes \frac{q}{d_k-1} \widetilde{b} \geq 0, \\
\operatorname{PT}&=\text{\(\mathds{1}\)}_{2} \otimes \frac{q}{d_k-1} \widetilde{b} - \text{\(\mathds{1}\)}_{2} \otimes p \widetilde{a} + \mathbb{I}_{2}\otimes p \widetilde{a} \geq 0,
\end{split}
\ee
where $\text{\(\mathds{1}\)}_{d_k}, \text{\(\mathds{1}\)}_{2}$ are identity matrices of dimensions $d_k$ and 2 respectively, $\mathbb{I}_{d_k}$ and $\mathbb{I}_{2}$ with all entries equal to 1 of dimensions $d_k$ and 2 respectively. To say that $\operatorname{PT}_{d_k}$ and $\operatorname{PT}$ are positive is enough to say that they have all eigenvalues $\lambda$ greater or equal to zero, so we can write:
\be
\begin{split}
\lambda\left(\operatorname{PT}_{d_k}\right)&=\lambda\left(\text{\(\mathds{1}\)}_{d_k} \otimes p \widetilde{a}\right)-\lambda\left(\text{\(\mathds{1}\)}_{d_k} \otimes \frac{q}{d_k-1} \widetilde{b}\right)+\lambda\left(\mathbb{I}_{d_k}\otimes \frac{q}{d_k-1} \widetilde{b}\right) \geq 0,\\
\lambda\left(\operatorname{PT}\right)&=\lambda\left(\text{\(\mathds{1}\)}_{2} \otimes \frac{q}{d_k-1} \widetilde{b}\right) - \lambda\left(\text{\(\mathds{1}\)}_{2} \otimes p \widetilde{a}\right) + \lambda\left(\mathbb{I}_{2}\otimes p \widetilde{a}\right) \geq 0.
\end{split}
\ee
Because $\operatorname{spec}\left(\mathbb{I}_{d_k}\right)=\left\{0,\ldots,0,d_k\right\}$, where 0 is taken $d_k-1$ times we have the following set of constraints
\be
\label{formulas}
\begin{split}
p\lambda\left(\widetilde{a}\right)+q\lambda(\widetilde{b}) \geq 0,\\
p\lambda\left(\widetilde{a}\right)-\frac{q}{d_k-1}\lambda(\widetilde{b})\geq 0,\\
\frac{q}{d_k-1}\lambda(\widetilde{b})+p\lambda\left(\widetilde{a}\right) \geq 0,\\
\frac{q}{d_k-1}\lambda(\widetilde{b})-p\lambda\left(\widetilde{a}\right)\geq 0.
\end{split}
\ee
Form the above we see that only nontrivial conditions are given by the second and fourth inequality, which are reduced (together) to equality
\be
\label{equality}
\frac{q}{d_k-1}\lambda(\widetilde{b})-p\lambda\left(\widetilde{a}\right)=0.
\ee
We see that to ensure $\operatorname{PPT}$ property, it is enough to satisfy only one constrain, which depends only on eigenvalues of submatrices of $\gamma_0$ and $\gamma_i$.
\section{Construction of the operators with specific constraints on spectra}
\label{XY}

In  the Section~\ref{II} we use a class of operators with the specific properties such that invariance respect to partial transposition on $B'$ system and particular spectra. Now, we present one of the possible realization of such operators. Namely, let us take (see~\cite{Badziag_pbits})
\be
\label{defXY}
X=\frac{1}{d_s\sqrt{d_s}}\sum_{i,j=1}^{d_s}u_{ij}|ij\>\<ji| , \ Y=\sqrt{d_s}X^{\operatorname{T}_{B'}}=\frac{1}{d_s}\sum_{i,j=1}^{d_s}u_{ij}|ii\>\<jj|,
\ee
where $u_{ij}$ are matrix elements of some unitary matrix $U\in \operatorname{M}(d_s \times d_s, \mathbb{C})$ with $|u_{ij}|=\frac{1}{\sqrt{d_s}}$.
It is easy to see that $\left(\text{\(\mathds{1}\)}_{B} \otimes \operatorname{T}_{B'} \right)X=X$ and $\left(\text{\(\mathds{1}\)}_{B} \otimes \operatorname{T}_{B'} \right)Y=Y$. Moreover, we can prove the following

\begin{fact}
\label{aux1}
Matrices $\sqrt{XX^{\dagger}}$ and $\sqrt{YY^{\dagger}}$, where $X,Y$ are given by the formula~\eqref{defXY} satisfy:
\be
\operatorname{spec}\left(\sqrt{XX^{\dagger}}\right)=\left\{\frac{1}{d_s^2},\ldots,\frac{1}{d_s^2}\right\},  \ \operatorname{spec}\left(\sqrt{YY^{\dagger}}\right)=\left\{\frac{1}{d_s},\ldots,\frac{1}{d_s},0,\ldots,0\right\},
\ee
where $d_s$ denotes dimension of the shield part, and for every matrix we have $d_s$ eigenvalues. Moreover, multiplicity of $1/d_s^2$ is equal to $d_s^2$, multiplicity of $1/d_s$ is equal to $d_s$ and finally, multiplicity of zeros is equal to $d_s(d_s-1)$.
\end{fact}

\begin{proof}
The proof is based on the following observation
\be
XX^{\dagger}=X^{\dagger}X=\frac{1}{d_s^4}, \quad YY^{\dagger}=Y^{\dagger}Y=\frac{1}{d_s^2}.
\ee
Let us redefine $X$ and $Y$ introducing $\widetilde{X}=d_s^2X$ and $\widetilde{Y}=d_s\sqrt{d_s}Y$.  We have  that$\widetilde{X}\widetilde{X}^{\dagger}=\widetilde{X}^{\dagger}\widetilde{X}=\text{\(\mathds{1}\)}$ and similarly for $\widetilde{Y}$.
Thanks to this we see that matrices $\widetilde{X}, \widetilde{Y}$ are unitary, so their eigenvalues are $\operatorname{e}^{\operatorname{i}\varphi_i}$, for $i=1,\ldots, d_s$. Now it is easy to deduce that
\be
\label{cc1}
\operatorname{spec}\left(X\right)=\left\{\frac{1}{d_s^2}\operatorname{e}^{\operatorname{i}\varphi_1},\ldots,\frac{1}{d_s^2}\operatorname{e}^{\operatorname{i}\varphi_{d_s}}\right\},  \ \operatorname{spec}\left(Y\right)=\left\{\frac{1}{d_s}\operatorname{e}^{\operatorname{i}\varphi_1},\ldots,\frac{1}{d_s}\operatorname{e}^{\operatorname{i}\varphi_{d_s}}\right\},
\ee
and
\be
\label{cc2}
\operatorname{spec}\left(XX^{\dagger}\right)=\left\{\frac{1}{d_s^4},\ldots,\frac{1}{d_s^4}\right\},  \ \operatorname{spec}\left(YY^{\dagger}\right)=\left\{\frac{1}{d_s^2},\ldots,\frac{1}{d_s^2}\right\}.
\ee
In equations~\eqref{cc1} and~\eqref{cc2} for simplicity we have omitted zeros in the spectra of $\operatorname{spec}(Y)$ and $\operatorname{spec}(YY^{\dagger})$. Moreover, they do not give us any nontrivial condition for positivity (see Section~\ref{tenMet}). Finally for $\sqrt{XX^{\dagger}}, \sqrt{YY^{\dagger}}$ we simply have to take square roots from every eigenvalue from above spectra to obtain the desired result.
\end{proof}

Now, in Lemma~\ref{depOnds} we can directly substitute $\sqrt{XX^{\dagger}}$ instead of $a_{kl}^{(0,0)}$, where $0 \leq k,l \leq d_k-1$ and $\sqrt{YY^{\dagger}}$ instead of $a^{(i,j)}_{mn}$, where $0 \leq m,n \leq 1$ and $0 \leq i,j \leq d_k-1$ for $i<j$ we obtain the specific example of pdit from our class.

\newpage
\bibliographystyle{plain}

\end{document}